\documentclass{jfrarticle}
\usepackage[utf8]{inputenc}
\usepackage[T1]{fontenc}
\usepackage{lmodern}
\usepackage{emptypage}
\usepackage{enumitem}
\if{
\newcommand{\myin}{1.3in}
\usepackage[
    top    = \myin,
    bottom = \myin,
    left   = \myin,
    right  = \myin]{geometry}
}\fi
\usepackage{setspace}
\usepackage[colorinlistoftodos,bordercolor=orange,backgroundcolor=orange!20,linecolor=orange,textsize=scriptsize]{todonotes}
\usepackage{amsmath}
\usepackage{amsfonts}
\usepackage{latexsym}
\usepackage{multirow}
\usepackage{hyperref}
\usepackage{xcolor}
\usepackage{mathtools, amssymb} 
\usepackage{phdalgo}
\usepackage{rotating}
\usepackage{listings}

\if{
\newtheorem{theorem}{Theorem}[chapter]
\newtheorem{lemma}[theorem]{Lemma}
\newtheorem{conjecture}[theorem]{Conjecture}

\newtheorem{property}[theorem]{Property}

\theoremstyle{definition}

\theoremstyle{remark}
\newtheorem{example}[theorem]{Example}
\newtheorem{remark}[theorem]{Remark}
}\fi
\newtheorem{theorem}{Theorem}%[section]
\newtheorem{conjecture}{Conjecture}%[section]

\newtheorem{lemma}[theorem]{Lemma}
\newdef{example}[theorem]{Example}
\newdef{remark}[theorem]{Remark}

\newcommand{\rational}[1]{{#1}^{\Q}}
\newcommand{\setA}{\mathcal{A}} % Semialgebraic functions
\newcommand{\setD}{\mathcal{D}} % Dictionnary of Univariate transcendental functions
\newcommand{\N}{\mathbb{N}}					% L'ensemble des entiers naturels
\newcommand{\Q}{\mathbb{Q}}					% Le corps des nombres rationnels
\newcommand{\R}{\mathbb{R}}					% Le corps des nombres reels
	
\newcommand{\etc}{\textit{etc}~}
\newcommand{\ie}{\textit{i.e.}~}
\newcommand{\eg}{e.g.~}

\newcommand{\resp}{resp.~}

\renewcommand{\geq}{\geqslant}
\renewcommand{\leq}{\leqslant}

\newcommand{\alphab}{\boldsymbol{\alpha}}

\newcommand{\fsa}{f_{\text{sa}}}
\newcommand{\fpop}{f_{\text{pop}}}
\newcommand{\epspop}{\epsilon_{\text{pop}}}

\newcommand{\Kpop}{K_{\text{pop}}}

\newcommand{\bop}{\mathtt{bop}}

\newcommand{\nbb}{\#{\text{boxes}}}

\numberwithin{equation}{section}

\DeclareMathOperator{\parab}{par}

\newcommand{\nlcertify}{\mathtt{NLCertify}}
\newcommand{\coq}{\text{\sc Coq}}

\newcommand{\sdpa}{\text{\sc sdpa}}
\newcommand{\hol}{\text{\sc Hol-light}}
\newcommand{\ocaml}{\text{\sc OCaml}}
\newcommand{\micromega}{\mathtt{micromega}}
\newcommand{\coqinterval}{\mathtt{interval}}

\newcommand{\astroot}{\mathtt{root}}

\newcommand{\composebop}{\mathtt{compose\_bop}}

\newcommand{\composeapprox}{\mathtt{compose\_approx}}
\newcommand{\buildpar}{\mathtt{build\_par}}

\newcommand{\templateapproxjfr}[3]{\mathtt{template\_approx} (#1, #2, #3)}
\newcommand{\templateoptim}{\mathtt{template\_optim}}
\newcommand{\templateapprox}{\mathtt{template\_approx}}

\newcommand{\minsa}{\mathtt{min\_sa}}
\newcommand{\maxsa}{\mathtt{max\_sa}}

\newcommand{\xb}{\mathbf{x}}

\newcommand{\zb}{\mathbf{z}}
 
\newcommand{\ab}{\mathbf{a}}
\newcommand{\bb}{\mathbf{b}}

\definecolor{darkgreen}{rgb}{0.0 0.5 0.0}
\definecolor{whitegreen}{rgb}{0.0 0.75 0.0}
\definecolor{whiteblue}{rgb}{0.0 0.0 1.2}
\definecolor{darkred}{rgb}{0.8 0.0 0.0}
\definecolor{dkblue}{rgb}{0,0.1,0.5}
\definecolor{lightblue}{rgb}{0,0.5,0.5}
\definecolor{dkgreen}{rgb}{0,0.4,0}
\definecolor{dk2green}{rgb}{0.4,0,0}
\definecolor{dkviolet}{rgb}{0.6,0,0.8}
\definecolor{dkpink}{rgb}{0.8,0,0.9}

\newcommand{\defcoq}{{{\tt:=}}}
\newcommand{\pluscoq}{{\tt+}}
\newcommand{\subcoq}{{\tt-}}

\newcommand{\eqcoq}{{\tt=}}
\newcommand{\evalexpr}[1]{[|#1|]_{\text{env}}}

\newcommand{\isos}{i_{\text{sos}}}
\newcommand{\sq}{\hat{\text{sq}}}
\newcommand{\itvsos}{itv}
\newcommand{\liftcoq}{\text{n}_{\text{lifting}}}

\newcommand{\iflst}{\if{\lstinline{$$}}\fi}
\newcommand{\code}[1]{\lstinline{#1}}

\title{Formal Proofs for Nonlinear Optimization}

\author{Victor Magron, Xavier Allamigeon, St\'ephane Gaubert and Benjamin Werner}
{VICTOR MAGRON\\Circuits and Systems Group,
Imperial College, London, UK\\
XAVIER ALLAMIGEON\\INRIA and CMAP, \'Ecole Polytechnique, CNRS, Palaiseau, France\\
ST\'EPHANE GAUBERT\\INRIA and CMAP, \'Ecole Polytechnique, CNRS, Palaiseau, France 
\and
BENJAMIN WERNER\\ LIX, \'Ecole Polytechnique, CNRS, Palaiseau, France}
\begin{abstract}
We present a formally verified global optimization framework. Given a semialgebraic or transcendental function $f$ and a compact semialgebraic domain $K$, we use the nonlinear maxplus template approximation algorithm to provide a certified lower bound of $f$ over $K$.
This method allows to bound 
in a modular way
some of the constituents of $f$ by suprema of quadratic forms with a well chosen curvature. Thus, we reduce the initial goal to a hierarchy of semialgebraic optimization problems, solved by sums of squares relaxations.  
Our implementation tool interleaves  semialgebraic approximations with sums of squares witnesses to form certificates. 
It is interfaced with {\sc Coq} and thus benefits from the trusted arithmetic available inside the proof assistant. 
 This feature is used to produce, from the certificates, both valid under-approximations and
lower bounds for each approximated constituent.
The application range for such a tool is widespread; for instance Hales' proof of Kepler's conjecture yields thousands of multivariate transcendental inequalities. We illustrate the performance of our formal framework on some of these inequalities as well as on 
examples from the global optimization literature.

\end{abstract}

\begin{document}
%sets the number of the first page
%\setcounter{page}{111} 
%\if{
\begin{bottomstuff}
The research leading to these results has received funding from the European Union's $7^{\text{th}}$ Framework Programme under grant agreement no. 243847 (ForMath).
\permission
\copyright\ 2013 Journal of Formal Reasoning
\end{bottomstuff}
\maketitle

\section{Introduction}
\label{sect:intro}
\subsection[Problems Involving Computer Assisted Proofs]{Problems Involving Computer Assisted Proofs}
This work is about one particular combination of secure formal proofs with fast mechanical computations: we want the computer to automatically determine precise numerical bounds of algebraic expressions, while retaining the safety of formal proofs - in our case, the $\coq$ proof system.

On one hand, since their conception, computers have been used as fast calculators,
whose speed allows the mathematician to access knowledge which would
be out of reach without the machine. On the other hand, it is now common
practice since a long while to use the computer as a rigorous {\em
  censor}, which checks the validity of a mathematical development
down to its formal steps in a given logical formalism; this is the
task of {\em proof systems} since the 1960s. Combining 
these two tasks, as in this work, is more recent. 

The programming language provided inside the formalism of $\coq$ can
be used in sophisticated ways. In particular, it allows to build
decision procedures or perform automatized reasoning, thus 
 to prove classes of propositions in a systematic and efficient fashion.
%In $\coq$, one can replace several deduction steps by a single computation step. 
Because this involves formalizing a fragment of the logic in $\coq$
itself, this technique is called {\em computational reflection}\index{computational reflection} and
was introduced in \cite{Barthe96atwo-level} (see also \cite{BM81} for details about reflection).

The fact that
complex computations can also take place inside the proof-system was
first used for some proof automation like the successive versions of
$\coq$'s ring tactic~\cite{Boutin97usingreflection,DBLP:conf/tphol/GregoireM05}   to check polynomial equalities
(actually we happen to use the current version of this tactic in the present work). Other applications include verifying large numbers' primality~\cite{pocklington}, checking witnesses from SAT/SMT solvers~\cite{smtcoq} or hardware verification~\cite{Pau96}.

Recently, proof-checkers embedded with computational features were
particularly highlighted by allowing the formal checking of results whose proofs are fundamentally computational, like the four-color theorem~\cite{Gonthier:2008:FCT:1425750.1425783} or Kepler's conjecture.

Kepler's conjecture is one of the eventual motivations of the present work. It can be stated as follows:

\begin{conjecture}[\theconjecture\ (Kepler 1611)]
\label{conj:kepler}
The maximal density of sphere packings in three dimensional space is $\pi / \sqrt{18}$.
\end{conjecture}

This conjecture has been proved by Thomas Hales\footnote{\url{https://code.google.com/p/flyspeck/wiki/AnnouncingCompletion}}. 

\begin{theorem}[\thetheorem\ (Hales~\cite{Hales94aproof,halesannmath})]
Kepler's conjecture is true.
\end{theorem}

One of the chapters of~\cite{halesannmath} is coauthored by Ferguson.
The publication of the proof, one of the ``most complicated [...] that has been ever produced'', to quote his author\footnote{\url{https://code.google.com/p/flyspeck/wiki/FlyspeckFactSheet}}, took several years and its verification
required ``unprecedented'' efforts by a team of referees; the difficulty being made worse by the use of mechanical computations interwound with mathematical deductions. The degree of complexity of such a checking has motivated the effort to fully formalize them.

Like the four-color theorem's proof, Hales' proofs thus combines ``conventional'' mathematical deduction and non-trivial computations. The formalization of this development is an ambitious goal addressed by the Flyspeck project, launched by Hales himself~\cite{hales:DSP:2006:432}. 
Note that other problems can be solved by proof assistants but do not rely on mechanical  computation. As an example, one can mention the formal proof of the  Feit-Thompson Odd Order Theorem~\cite{oddorderth}. Flyspeck also involves the formalization of many mathematical  concepts and proofs; but we here do not deal with the ``conventional'' mathematical part of the project.

\subsection{Nonlinear Inequalities} 
\label{subsec:introfly}
% The Flyspeck project~\cite{hales:DSP:2006:432}, 
% %directed by Thomas Hales, %SG already said a few lines above
% %BW : agree with SG.
% is a large-scale effort  to complete the formal verification of % Kepler's conjecture. In Flyspeck, various mathematical tools need to % be tackled in a formal setting.

\if{Computations are mandatory for at least three kind of tasks in Hales' proof: handling the generation of some special planar graphs, producing bounds for linear expressions, and for non-linear expressions.
\todo[inline]{BW : est-ce-que cette formulation est correcte ? Alternative possible : \dots planar graphs, use of linear programming, and bounding of non-linear expressions.}
}\fi
Computations are mandatory for at least three kind of tasks in Hales' proof: generation of planar graphs, use of linear programming, and bounding of non-linear expressions.
% Extensive computation are mandatory to handle the formal generation of % some special planar graphs. The formal proofs of the bounds of linear % and nonlinear programs also require major computational time. 
% \todo[inline]{SG: I dont get what ``programs'' means here, % inequalities?}
Details about the two former issues are available in Solovyev's doctoral dissertation~\cite{pittir16721}.

We here focus on the last issue, namely the formal checking of the correctness of hundreds of nonlinear inequalities.
Each of these cases boils down to the computation of a certified lower bound for a real-valued multivariate
function $f : \R^n \to \R$ over a compact semialgebraic set
$K \subset \R^n$:
%Our aim is to automatically provide lower bounds for the following
%global optimization problem:
\begin{equation}
\label{eq:f}
f^*  :=  \inf_{\xb \in K} f (\xb) \enspace,
\end{equation}
In some cases, $f$ will be a multivariate
polynomial (polynomial optimization problems (POP)); alternatively, $f$ may belong to the
algebra $\setA$ of semialgebraic functions, which extends
multivariate polynomials, obtained through arbitrary compositions of
$(\cdot)^{p}$, $(\cdot)^{\frac{1}{p}} (p \in \N_{>0})$, 
$\lvert\cdot\rvert$, $+$, $-$, $\times$, $/$, $\sup(\cdot,\cdot)$,
$\inf(\cdot,\cdot)$ (semialgebraic optimization problems);
finally, in the most general case, $f$ may, in addition, involve transcendental functions ($\sin$, $\arctan$, \etc).

Our aim is twofold:
\begin{itemize}
\item The first is {\em automation}: we want to design a method that finds sufficiently precise lower bounds for all, or at least a majority of the functions $f$ and domains $K$ occurring in the proof.
\item The second, as already stressed above, is {\em certification}; meaning that the correctness of each of these bounds must be, eventually, formally provable in a proof system such as $\coq$.
\end{itemize}
%One among many applications is the set of several
%thousands of non-linear inequalities which occur in Thomas Hales'
%proof of Kepler's conjecture, which is formalized in the Flyspeck
%project~\cite{DBLP:journals/dcg/HalesHMNOZ10,hales:DSP:2006:432}. 
%Several
%inequalities issued from Flyspeck actually deal with special cases
%of Problem~\eqref{eq:f}. 

An additional crucial point is {\em precision}. Especially the inequalities of Hales' proof are essentially tight.

%Formal methods that produce precise bounds are mandatory because 
%the inequalities to be proved are essentially tight.
%of the
% tightness of these inequalities. 

However, the application range for formal bounds reaches over many areas, way beyond the proof of Kepler Conjecture. Hence, we are also keen on tackling {\em scalability} issues, which arise when one wants to provide coarser lower bounds for optimization problems with a larger number of variables or polynomial inequalities of a higher degree. 

\subsection{Context}
There have been a number of related efforts to obtain formal proofs for global optimization.

Lower bounds for POP can be obtained by solving sums of squares (SOS) programs
using the numerical output of specialized semidefinite programming (SDP) software~\cite{henrion:hal-00172442}. 
Such techniques rely on hybrid symbolic-numeric
certification methods, see Peyrl and Parrilo~\cite{DBLP:journals/tcs/PeyrlP08} and Kaltofen et al.~\cite{KLYZ09}, 
which in turn allow to produce nonnegativity certificates which can be checked in proof assistants.
Related formal frameworks include a decision procedure in $\coq$, described in~\cite{Besson:2006:FRA:1789277.1789281} as well as in $\hol$~\cite{harrison-sos}. Both procedures include a proof-search step to find nonnegativity certificates, which relies on the same $\ocaml$ libraries. The procedure in~\cite{Besson:2006:FRA:1789277.1789281} is implemented as a tactic called \code{micromega}.

Alternative approaches to SOS are based on formalizing multivariate Bernstein polynomials. This research has been carried out in the thesis of R. Zumkeller~\cite{Zumkeller:2008:Thesis} and by Mun\~oz and Narkawicz~\cite{BernsteinPVS} in PVS~\cite{cade92-pvs}.

Such polynomial optimization methods can be extended to transcendental functions using multivariate polynomial approximation through a semialgebraic relaxation. 
This requires to 
%implies being
be able to also certify the approximation error in order to conclude. MetiTarski~\cite{Akbarpour:2010:MAT:1731476.1731498} is a theorem prover that can handle nonlinear inequalities involving special functions such as $\ln$, $\cos$, etc. These univariate transcendental functions (as well as the square root) are approximated by a hierarchy of approximations which are rational functions derived from Taylor or continued fractions expansions (for more details, see Cuyt et al.~\cite{cuyt2008handbook}).

The Flyspeck project also employed specific methods to verify nonlinear inequalities. Hales and Solovyev developed a nonlinear verification framework~\cite{pittir16721,DBLP:journals/corr/abs-1301-1702}, mixing $\ocaml$ and $\hol$~\cite{DBLP:conf/fmcad/Harrison96} procedures  to achieve formal Taylor interval approximations. A part of this procedure is informal and aims to provide useful hints such as an appropriate subdivision of the nonlinear inequality box $K$. The formal part of the procedure uses formalization results related to the multivariate Taylor Theorem (e.g. multivariate Taylor formula with second-order error terms) and formal interval arithmetic. Numerical computations with finite precision floating-point numbers are done in a formal setting within $\hol$, thanks to a careful representation of natural numerals over arbitrary bases (see~\cite{DBLP:journals/corr/abs-1301-1702} for more details). This formal framework is about 2000$\sim$4000 times slower than an informal procedure (written in C++) performing the same verification. 
%Currently, it takes about 5000 CPU hours to verify all the nonlinear inequalities, using formal Taylor interval approximations~\cite{phdsolovyev} in the $\hol$ proof assistant~\cite{DBLP:conf/fmcad/Harrison96}.
%Recently, all Flyspeck nonlinear inequalities could be verified using Taylor interval approximations~\cite{DBLP:journals/corr/abs-1301-1702} in the $\hol$ proof assistant~\cite{DBLP:conf/fmcad/Harrison96}.  

 In~\cite{Chevillard11cheb}, the authors present a scheme amenable to formalization, which provides certified polynomial approximations of univariate transcendental functions. An upper bound of the approximation error is obtained by using a second approximation polynomial with bounded approximation, error relying on a non-negativity test performed by means of univariate sums of squares.
%We also mention procedures that solve SMT problems over the real numbers, using interval constraint propagation\cite{DBLP:journals/corr/abs-1204-3513}. 
The Flocq  library~\cite{BM11Flocq} formalizes floating-point arithmetic inside $\coq$. The tactic $\coqinterval$~\cite{Melquiond201214}, built on top of Flocq, can simplify inequalities on expressions of real numbers. Our formal framework relies on this tactic for handling univariate transcendental functions.
However, inequalities involving multivariate transcendental functions remain typically
difficult to solve with interval arithmetic, in particular due to the
correlation between arguments of unary functions (e.g. $\sin, \arctan$)
or binary operations (e.g. $ +, -, \times, /$).

Currently, it takes about 5000 CPU hours to verify all the nonlinear inequalities with formal Taylor interval approximations (developed by the Flyspeck project) in $\hol$. One motivation of the present work is to reduce this total verification time using alternative formal methods.
In a previous work, the authors developed an informal framework, built on top of the certified --template based -- global optimization method~\cite{victormpb,victorcicm, victorecc}. The nonlinear template method is a certification framework, aiming at handling the approximation of transcendental functions and increasing the size of certifiable instances. It combines the ideas of maxplus approximations~\cite{a5,a6} and linear templates~\cite{Sankaranarayana+others/05/Scalable} to reduce the complexity of the semialgebraic approximations. Given a multivariate transcendental function $f$ and a semialgebraic compact set $K$, one builds lower semialgebraic approximations of $f$ using maxplus approximations (usually a supremum of quadratic forms) choosing a set of control points. In this way, the nonlinear template algorithm builds a hierarchy of semialgebraic relaxations that are solved with SDP.

\subsection{Contributions}
In this article, we present a formal framework, built on top of this informal method. %certified --template based -- global optimization method~\cite{victorcicm, victorecc}. %The template optimization algorithm 

The correctness of the bounds for semialgebraic optimization problems can be verified using the interface of this algorithm with the $\coq$ proof assistant. Thus, the certificate search and the proof checking inside $\coq$ are separated, which is common in the so-called {\em sceptical} approach~\cite{Barendregt02}. There are some more practical difficulties however.
When solving semialgebraic optimization problems (e.g. POP), the sums
of squares certificates produced by existing tools do not exactly match with the system of
polynomial inequalities defining $K$, because these external tools use limited precision floating point numbers and are thus prone to rounding errors. A
certified upper bound of this error is obtained inside the proof
assistant. Once the bounding of the error is obtained, the
verification of the certificate is performed through an equality check
in the ring of polynomials whose coefficients are arbitrary-size
rationals. This means that we benefit from efficient arithmetic of these coefficients: the recent implementation of functional modular arithmetic allows to handle arbitrary-size natural numbers~\cite{GT:ijcar06}. Spiwack~\cite{spiwack2006} has modified the virtual machine to handle 31-bits integers natively, so that arithmetic operations are delegated to the CPU. These recent developments made possible to deal with cpu-intensive tasks such as handling the proof checking of SAT traces~\cite{armand2010}. Here, it allows to check efficiently the correctness of SOS certificates.
%The $\coq$ proof assistant can perofEfficient computation inside the proof assistant. 
%Certified interval enclosures  
%This formal method framework extends to more general classes of semialgebraic problems 
Furthermore, this verification for SDP relaxations is combined to deduce the correctness of semialgebraic optimization procedures, which requires in particular to assert that the semialgebraic functions are well-defined. It allows to handle more complex certificates for non-polynomial problems.
Finally, the datatype structure of these certificates allows to reconstruct the steps of the nonlinear template optimization algorithm.

The present paper is a followup of~\cite{victorcicm} in which the idea of maxplus approximation of transcendental functions was improved through the use of template abstractions. Here, we develop and detail the formal side of this approach. The present framework provides an automated decision procedure to obtain formal bounds for polynomial and semialgebraic functions over semialgebraic sets. This formalization is associated with the development of $\coq$ libraries within the software package~\texttt{NLCertify} (see~\cite{icms14} as well as the software web-page\footnote{\url{http://nl-certify.forge.ocamlcore.org/}}) and complements the existing libraries of the software, originally written in $\ocaml$.

The paper is organized as follows. 
Section~\ref{sect:formalpop} is devoted to formal polynomial optimization.
We recall some properties of SDP relaxations for
polynomial problems (Section~\ref{subsect:certpop}). In
Section~\ref{subsect:hybrid}, we outline the conversion of the numerical
SOS produced by the SDP solvers into an exact rational
certificate. Section~\ref{subsect:coqpop} describes the formal verification of this certificate inside $\coq$.
Section~\ref{sect:formalfsa} explains  how to reduce semialgebraic problems to POP through the Lasserre-Putinar lifting. The structure of the interval enclosure certificates for semialgebraic functions is described in Section~\ref{sect:coqfsa}.
We remind the principle of the nonlinear maxplus template method in Section~\ref{subsect:template}. The interface between this algorithm and the formal framework is presented in Section~\ref{subsect:interface}.
%Benchmarks illustrate the efficiency of the formal global optimization method, by certifying bounds of non-linear problems involving up to $10^3$ variables, as well as non trivial inequalities issued from the Flyspeck project.
Finally, we demonstrate the scalability of our formal method by certifying bounds of 
non-linear problems from the global optimization literature as well as
non trivial inequalities issued from the Flyspeck project.

\if{
The max-plus approximation, and the main algorithm based on the non-linear templates method are presented in Section~\ref{sect:template}. 
Numerical results are presented in Section~\ref{sect:bench}. We demonstrate the scalability of our approach by certifying bounds of 
non-linear problems involving up to $10^3$ variables, as well as
non trivial inequalities issued from the Flyspeck project.
}\fi

\section{Formal Polynomial Optimization}
\label{sect:formalpop}

We consider the general constrained polynomial optimization problem (POP):
\begin{equation}
\label{eq:cons_pop}
\fpop^*  :=  \inf_{\xb \in \Kpop} \fpop (\xb),
\end{equation}
where $\fpop : \R^n \to \R$ is a $d$-degree multivariate polynomial,
$\Kpop$ is a semialgebraic compact set defined by inequality
constraints $g_1(\xb) \geq 0, \dots, g_m(\xb) \geq 0$,  where
$g_j(\xb) : \R^n \to \R$ is each time a real-valued polynomial.
Recall that a $d$-degree multivariate polynomial $p$ can be decomposed as $p(\xb) = \sum_{|\alphab| \leq d} p_{\alphab} \xb^{\alphab}$, where each $\alphab$ is a nonnegative integer vector $(\alpha_1,\dots, \alpha_n)$, with $|\alphab| := \sum_{i=1}^n \alpha_i$.

Since the domain $\Kpop$ is compact, we know that it is included in some
box, say  $[\ab, \bb] := [a_1, b_1] \times \dots \times [a_n, b_n]
\subset \R^n$. We can thus assume, without loss of generality, that
the first constraints are precisely box constraints; more precisely,
that $m \geq 2 n$ and $g_1 := x_1 - a_1, g_2 := b_1 - x_1, \dots,~g_{2 n -1} := x_n - a_n, ~g_{2 n} := b_n - x_n$. 
In practice, such bounds $[\ab, \bb]$ are known in advance for all Flyspeck inequalities as well as for other global optimization problems which have been considered here.

Recall that the {\em set of feasible points} of an optimization
problem is simply the domain over which the optimum is taken,
i.e., here, $\Kpop$.
% the {\em feasible set} of Problem~\eqref{eq:cons_pop}.

\subsection{Certified Polynomial Optimization using SDP Relaxations} %Lasserre's hierarchy of semidefinite relaxations
\label{subsect:certpop}

Here, we remind how to cast a POP into an SOS program, which can be in turn written as an SDP. We define the
set of polynomials which can be written as a sum of squares $\Sigma[\xb] := \Bigl\{\,\sum_i q_i^2, \, \text{ with } q_i \in \R[\xb] \Bigr\}$. We set $g_0 := 1$ and take $k \geq k_0 := \max(\lceil d / 2 \rceil , \lceil \deg g_1 / 2 \rceil, \dots, \lceil \deg g_m / 2 \rceil)$.
Then, we consider the following
hierarchy of SDP relaxations for Problem~\eqref{eq:cons_pop},
consisting of the optimization problems $Q_k$ over the variables $(\mu, \sigma_0, \dots, \sigma_m)$:
\[
Q_k:\left\{			
\begin{array}{ll}
\sup_{\mu, \sigma_0, \dots, \sigma_m} & \mu \\			 
\text{s.t.} & \fpop (\xb) - \mu = \sum_{j = 0}^m \sigma_j(\xb) g_j(\xb), \ \forall \xb \in \R^n, \enspace \\
& \mu\in \R, \ \sigma_j \in \Sigma[\xb], \deg(\sigma_j g_j) \leq 2 k, \  j = 0,\dots,m \enspace.
\end{array} \right. \]
The integer $k$ is called the {\em SDP relaxation order} and  $\sup
(Q_k)$ is the optimal value of $Q_k$. A feasible point $(\mu_k,\sigma_0,\dots,\sigma_m)$ of Problem $Q_k$ is said to be an {\em SOS certificate},
showing the implication $g_1(\xb)\geq 0,\dots,g_m(\xb) \geq 0 
\implies \fpop(\xb)\geq \mu_k$. 
%We also refer to the polynomial $\sum_{j = 0}^m \sigma_j(\xb) g_j(\xb)$ as a {\it Putinar-type certificate} since it comes from the representation theorem of positive polynomials by Putinar~\cite{putinar1993positive}.
We also use the term  {\it Putinar-type certificate} since its
existence comes from the representation theorem of positive
polynomials by Putinar~\cite{putinar1993positive}. 

The sequence of optimal values $(\sup (Q_k))_{k \geq k_0}$ is monotonically increasing.
Las\-serre showed \cite{DBLP:journals/siamjo/Lasserre01} 
that it does converge to $\fpop^*$ under an additional assumption
on the polynomials $g_1,\dots,g_m$ (see~\cite{DBLP:journals/siamjo/Schweighofer05} for more details).  
%Here, we will consider sets $\Kpop$ included in a box of $\R^n$.
One way to ensure that this assumption is automatically satisfied is to normalize and index the box inequalities as follows (corresponding to the affine transformation $x_i \mapsto (x_i - a_i) / b_i, i = 1, \dots, n$):
%After normalization and indexing the box inequalities as follows:
\begin{equation}
\label{eq:norm_box}
g_1 (\xb) := x_1, \: g_2 (\xb) := 1 - x_1, \dots, g_{2 n - 1} (\xb) := x_n, \: g_{2 n } (\xb) := 1 - x_n \enspace, 
\end{equation}
%one way to ensure that this assumption is automatically satisfied is to add 
then to add the redundant constraint $n - \sum_{j = 1}^n {x_j^2} \geq 0$ to the set of constraints. For the sake of simplicity, we assume that the inequality constraints of Problem\eqref{eq:cons_pop} satisfy both conditions.
% the sequel that the inequality constraints of Problem\eqref{eq:cons_pop} satisfy~\eqref{eq:norm_box} and the 

Also note that our current implementation allows to compute lower bounds for POP more efficiently by using a sparse refinement of the hierarchy of SDP relaxations $(Q_k)$ (see~\cite{Waki06sumsof} for more details).

\subsection{Hybrid Symbolic-Numeric Certification}
\label{subsect:hybrid}
The general scheme is thus quite clear: an external tool, here acting
as an oracle, computes the
certificate $(\mu_k, \sigma_0,\dots,\sigma_m)$ and the formal proof
essentially boils down to checking the equality 
\[
\fpop (\xb) - \mu_k = \sum_{j = 0}^m \sigma_j(\xb) g_j(\xb) \enspace
\]
and $\coq$'s ring tactic can typically verify such equalities.

There are practical difficulties however. 
In practice, we solve the relaxations $Q_k$ using SDP solvers (\eg{}$\sdpa$ \cite{YaFuNaNaFuKoGo:10}). Unfortunately, such
solvers are implemented using floating-point arithmetic and the solution $(\mu_k, \sigma_0,\dots,\sigma_m)$  satisfies only approximately the equality constraint in $Q_k$:
\[
\fpop (\xb) - \mu_k \simeq \sum_{j = 0}^m \sigma_j(\xb) g_j(\xb) \enspace.
\]

More precisely, the optimization problems are formalized in $\coq$ by
using rational numbers for the coefficients. In any case, we need to
deal with this approximation error.

An elaborate method would be to obtain exact certificates, for
instance by the rationalization scheme (rounding and projection
algorithm) developed by Peyrl and
Parrilo~\cite{DBLP:journals/tcs/PeyrlP08}, with an improvement of
Kaltofen et al.~\cite{KLYZ09}.  Let us note $ \theta_k := \Vert \fpop
(\xb) - \mu_k - \sum_{j = 0}^m \sigma_j(\xb) g_j(\xb) \Vert$ the error
for the problem $Q_k$.  The method of Kaltofen et al. ~\cite{KLYZ09}
consists in applying first Gauss-Newton iterations to refine the
approximate SOS certificate, until $\theta_k$ is less than a given
tolerance and then, to apply the algorithm
of~\cite{DBLP:journals/tcs/PeyrlP08}.  The number $\mu_k$ is
approximated by a nearby rational number $\rational{\mu_k} \lessapprox
\mu_k$ and the approximate SOS certificate $(\sigma_0, \dots
,\sigma_m)$ is converted to a rational SOS (for more details,
see~\cite{DBLP:journals/tcs/PeyrlP08}).  Then the refined SOS is
projected orthogonally to to the set of rational SOS certificates
$(\rational{\mu_k},\rational{\sigma_0}, \dots,\rational{\sigma_m})$,
which satisfy (exactly) the equality constraint in $Q_k$.  This can be
done by solving a least squares problem,
see~\cite{DBLP:journals/tcs/PeyrlP08} for more information.  Note that
when the SOS formulation of the polynomial optimization problem is not
strictly feasible, then the rounding and projection algorithm may
fail. However, Monniaux and Corbineau proposed a partial workaround
for this issue~\cite{Monniaux_Corbineau_ITP2011}.  In this way, except
in degenerate situations, we arrive at a candidate SOS certificate
with rational coefficients, $(\rational{\mu_k},\rational{\sigma_0},
\dots,\rational{\sigma_m})$ from the floating point solution of
$(Q_k)$.

In our case, we do not use the rounding and projection algorithm of
Peyrl and Parrilo; instead we rely on a simpler and cruder
scheme. We perform a certain number of operations before handing over
the certificate to $\coq$. 

In practice, the SDP solvers solve an optimization problem (equivalent to $Q_k$) over symmetric matrix variables $Z_0, \dots, Z_m$. From any floating point solution of this equivalent problem, one can extract the vectors $\mathbf{v}_{ij}$ of $Z_j$ with the associated $r_j$ coefficients $(\lambda_{ij})_{1 \leq i \leq r_j}$. Let $v_{i j}$ be the polynomial with vector coefficient $\mathbf{v}_{ij}$. Then, one has the following decomposition:
\begin{equation}
\label{eq:sosdecomp}
\sigma_j (\xb) = \sum_{i = 1}^{r_j} \lambda_{i j} v_{i j}^2(\xb), j = 0, \dots, m \enspace.
\end{equation}
The extraction is done with the {\sc Lacaml} (Linear Algebra with $\ocaml$) library, implementing the {\sc Blas}/{\sc Lapack}-interface.
The floating-point numbers of the generated certificate are viewed as rationals through the
straightforward mapping. Numerical SOS certificates are converted into rational SOS using the function $\mathtt{Q.of\_float}$ of the {\sc Zarith} $\ocaml$ library, which implements arithmetic and logical operations over arbitrary-precision integers. The floating-point value $\mu_k$ is also converted into a rational.

Then, we compute the (exact) error polynomial:
$$\epspop(\xb) := \fpop(\xb) - {\mu}_k -  \sum_{j = 0}^m  {\sigma}_j(\xb) g_j(\xb) \enspace .$$ 
Now we explain how to provide another, hopefully small, bound for $\epspop$. Note that this polynomial can be decomposed as 
$\epspop(\xb) = \sum_{|\alphab| \leq 2 k} \epsilon_{\alphab} \xb^{\alphab}$. Fortunately, the coefficients of this polynomial are
generally small, which allows us to choose  
$$\epspop^*  := \sum_{\epsilon_{\alphab} \leq 0} \epsilon_{\alphab}.$$
Indeed, the box inequalities guaranty that for each $\xb \in [0, 1]^n$:
\begin{equation}
\label{eq:epspop}
\epspop (\xb) =\sum_{|\alphab| \leq 2 k} \epsilon_{\alphab} \xb^{\alphab}
 ~\leq ~ \sum_{\epsilon_{\alphab} \leq 0} \epsilon_{\alphab} = \epspop^* \enspace.
\end{equation}

Finally, we compute the actual exact bound given by the certificate:
$\mu_k^- := \mu_k + \epspop^*$.
We see that  $\mu_k^- := \mu_k + \epspop^*$ is a valid lower bound of $\fpop$ over the domain $\Kpop$.

Note that we could optimize the polynomial $\epspop$ over $[0, 1]^n$, but it would be as hard as solving the initial POP. Moreover, one would have to consider again some residual polynomial after solving the corresponding SDP relaxation.

\subsection{A Formal Checker for Polynomial Systems}
\label{subsect:coqpop}
%\label{sec:polnonneg}
Following the procedure described in Section~\ref{subsect:hybrid}, we extract a rational certificate $(\mu_k^-, \sigma_0,\dots, \sigma_m, \epspop)$. 

By definition, this certificate satisfies the following, for all $\xb \in [0, 1]^n$:
\begin{equation}
\label{eq:ringpop}
 \fpop(\xb) - \mu_k^- =   \sum_{j = 0}^m  \sigma_j(\xb) g_j(\xb) + (\epspop(\xb) - \epspop^*)\enspace.
\end{equation}

%\subsubsection{Encoding Polynomials and SOS certificates in $\coq$}
The procedure which checks the equality~\eqref{eq:ringpop} between
polynomials and SOS inside $\coq$ relies on  computational
  reflexion. We use the reflexive \code{ring} tactic, by using a so-called ``customized'' polynomial ring \footnote{\url{http://coq.inria.fr/refman/Reference-Manual029.html}}. 
Given two polynomials $p$ and $q$, this tactic verifies the polynomial equality ``$p=q$'' in two steps. The first step is a normalization of both $p$ and $q$ w.r.t. associativity, commutativity and distributivity, constant propagation and rewriting of monomials. The second step consists in comparing syntactically the results of this normalization.
%an efficient implementation \code{bigQ} of arbitrary-size rational numbers (see )

Given a sequence of polynomial constraints 
\lstinline{g $\defcoq$  $[g_1, \dots, g_m]$}, a lower bound $\mu_k^-$, an objective
polynomial $\fpop$ and a POP certificate \code{cert_pop} (build with a
Putinar-type certificate and a polynomial remainder $\epspop$),
the fact that a successful check of the certificate entails
nonnegativity of the polynomial is formalized by the following
correctness lemma: 

\begin{lstlisting}
Lemma correct_pop env g $\fpop$ cert_pop $\mu_k^-$:  
g_nonneg env g $\to$ checker_pop g $\fpop$ $\mu_k^-$ cert_pop  $\eqcoq$ true $\to$ 
$\mu_k^-$ $\leq$ $\evalexpr{\fpop}$.
\end{lstlisting}
The way this lemma is stated shows that in our development, we use an environment function \lstinline{env} to
bind positive integers to polynomial real variables. The function \lstinline{$\evalexpr{\cdot}$} maps a polynomial expression to the carrier type \code{R}. 
The function
\code{g_nonneg} explicits the conditions by returning the conjunction of propositions $\evalexpr{g_1} \geq 0 \wedge \dots \wedge \evalexpr{g_m} \geq 0$.

In the sequel of this section, we describe the data structure and the auxiliary lemmas that allow to define and prove~\code{correct_pop}.
\subsubsection{Encoding Polynomials}
Checking ring equalities between polynomials requires to provide a type of coefficients. In our current setting, we choose \code{bigQ}, the type of arbitrary-size rationals. The ring morphism \code{IQR} injects rational coefficients into the carrier type \code{R} of $\coq$ classical real numbers. For the sequel, we also note:
\begin{lstlisting}
Notation "[ c ]" $\defcoq$ IQR c.
\end{lstlisting} 
We use two types of polynomials: \code{PExpr} is for uninterpreted ring expressions while \code{PolC} is for uninterpreted normalized polynomial expressions:

\begin{lstlisting}
 Inductive PExpr : Type $\defcoq$
  | PEc   : bigQ $\to$ PExpr
  | PEX   : positive $\to$ PExpr
  | PEadd : PExpr $\to$ PExpr $\to$ PExpr
  | PEsub : PExpr $\to$ PExpr $\to$ PExpr
  | PEmul : PExpr $\to$ PExpr $\to$ PExpr
  | PEopp : PExpr $\to$ PExpr
  | PEpow : PExpr $\to$ N $\to$ PExpr.
\end{lstlisting} 

\begin{lstlisting}
 Inductive PolC : Type $\defcoq$
  | Pc   : bigQ $\to$ PolC
  | Pinj : positive $\to$ PolC $\to$ PolC
  | PX   : PolC $\to$ positive $\to$ PolC $\to$ PolC.
\end{lstlisting}
The three constructors \code{Pc}, \code{Pinj} and \code{PX} satisfy the following conditions:
\begin{enumerate}
\item The polynomial \code{(Pc c)} is the constant polynomial that evaluates to \code{[c]}.  

\item The polynomial \code{(Pinj i p)} is obtained by shifting the index of \code{i} in the variables of \code{p}. In other words, when \code{p} is interpreted as the value of the ($n - i$) variables polynomial $p(x_1, \dots, x_{n - i})$, then one interprets \code{(Pinj i p)} as the value of $p(x_i, \dots, x_n)$.

\item Let \code{p} (\resp{}\code{q}) represents $p$ (\resp{}$q(x_1, \dots, x_{n - 1})$). Then \code{(PX p j q)} evaluates to $p x_1^j + q(x_2, \dots, x_n)$.
\end{enumerate}
%}\fi
Polynomial expressions can be normalized  via the procedure \lstinline{norm : PExpr  $\to$ PolC}. We note \lstinline{$p_1 \equiv p_2$} the boolean equality test between two normal forms \lstinline{$p_1$} and \lstinline{$p_2$}.
%The function \code{eval_pol} (\resp{}\code{eval_expr}) maps a sparse Horner polynomial (\resp{}a polynomial expression) to the carrier type \code{R}. Both procedures rely on a specific environment function \lstinline{env $\defcoq$ positive $\to$ R}, which binds positive integers to the polynomial real variables, \eg{}\code{(env 1)} returns \code{x1}, \code{(env 2)} returns \code{x2}.

%It relies on a specific environment function \lstinline{env $\defcoq$ positive $\to$ R}, which binds positive integers to the polynomial real variables, \eg{}\code{(env 1)} returns \code{x1}, \code{(env 2)} returns \code{x2}.
\if{\lstinline{$$}}\fi
The procedure \code{checker_pop} for SOS certificates relies on the correctness lemma \code{norm_eval}:
\begin{lstlisting}
Lemma norm_eval ($p$ : PExpr) ($q$ : PolC) : 
norm($p$) $\equiv$ $q$ $\implies$ $\forall$ env, $\evalexpr{p}$ $\eqcoq$ eval_pol env $q$.
\end{lstlisting}
Here, the function \code{eval_pol} maps a sparse normal form to the carrier type \code{R}.
\subsubsection{Encoding SOS certificates}

%\begin{lstlisting}
%Fixpoint g_nonneg env g $\defcoq$ match g with 
%  | [::]   $\implies$ True  
%  | $g_j$::tl $\implies$ $\evalexpr{g_j}$ $\geq$ 0  $\wedge$ g_nonneg env tl
%  end.
%\end{lstlisting}
We recall that an SOS can be decomposed as $\sigma_j := \sum_{i = 1}^{r_j} \lambda_{i j} v_{i j}^2$.
Each $\sigma_j$ is encoded using a finite sequence of tuples composed of an arbitrary-size rational (of type \code{bigQ}) and a  polynomial (of type \code{PolC}). Then, we build Putinar-type certificates $\sum_{j = 0}^m  \sigma_j(\xb) g_j(\xb)$ of type \code{cert_putinar}, with a finite sequence of tuples, composed of an SOS $\sigma_j$ and a polynomial $g_j$.
Finally, we define POP certificates (object of type \code{cert_pop}) for Problem~\ref{eq:cons_pop} using a Putinar-type certificate and a polynomial remainder $\epspop$.
%Any rational coefficient of the sequence $\lambda := [\lambda_{1 1}; \dots; \lambda_{1 r_m}]$ (\resp{}\lstinline{g}) can be retrieved using an integer index. 
%\begin{lstlisting}
%Definition $\sigma$ $\defcoq$ seq  ( $\timescoq$ PolC). 
%Definition cert_putinar $\defcoq$ seq ($\sigma$ $\timescoq$ nat).
%\end{lstlisting}

\if{
the record construction:
\begin{lstlisting}
Record cert_pop $\defcoq$ mk_cert_pop {$\epspop$: PolC;s: cert_putinar}.
\end{lstlisting}
}\fi
%$\lambda$ : seq bigQ;

\subsubsection{Formal proofs for polynomial bounds}
The coarse lower bound $\epspop^*$ of a polynomial remainder $\epspop$ can be computed inside $\coq$ with the following recursive procedure:
\begin{lstlisting}
Fixpoint lower_bound_0_1 $\epspop$ $\defcoq$ match $\epspop$ with
  | Pc c     $\implies$ min c 0
  | Pinj _ p $\implies$ lower_bound_0_1 p
  | PX p _ q $\implies$ lower_bound_0_1 p $\pluscoq$ lower_bound_0_1 q
end.
\end{lstlisting}
The remainder inequality~\eqref{eq:epspop} can then be proved by structural induction. 

\if{Given a list of hypotheses \code{g}, a lower bound $\mu_k^-$, an objective polynomial $\fpop$ and a POP certificate \code{cert_pop} (build with an SOS certificate \code{sos}, a polynomial \lstinline{$\epspop$} and a sequence of arbitrary-size rationals \lstinline{$\lambda$}) }\fi
Any certificate \code{c} of type \code{cert_putinar} can be mapped to a sparse Horner form with the function \code{toPolC} function, using the sequences \code{g}, \lstinline{$\lambda$} and the environment \lstinline{env}. Since each $g_j$ is nonnegative by assumption and each $\sigma_j$ is an SOS, one can verify easily the nonnegativity of a Putinar-type certificate by checking the nonnegativity of each \lstinline{$\lambda$} element. 
The boolean function \code{checker_pop} verifies that:
\begin{enumerate}
\item each  element of $\lambda$ is nonnegative
\item  \lstinline{norm($\fpop - \mu_k^-$) $\equiv$  (toPolC g $\lambda$ sos) $\pluscoq$ $\epspop$ $\subcoq$ [lower_bound_0_1 $\epspop$]} \newline (\ie{}equality~\eqref{eq:ringpop} is satisfied)
\end{enumerate}

\if{
Then, polynomial nonnegativity can be verified by applying the following correctness lemma: 
\begin{lstlisting}
Lemma correct_pop env g $\fpop$ cert_pop $\mu_k^-$:  
g_nonneg env g $\to$ checker_pop g $\fpop$ $\mu_k^-$ cert_pop  $\eqcoq$ true $\to$ 
$\mu_k^-$ $\leq$ $\evalexpr{\fpop}$.
\end{lstlisting}
}\fi
Then, we can prove Lemma \code{correct_pop}:
\begin{proof}
By assumption, one can apply \code{norm_eval} with \lstinline{$p = \fpop - \mu_k^-$} and \lstinline{$q$ $\eqcoq$ (toPolC g $\lambda$ sos) $\pluscoq$ $\epspop$ $\subcoq$ [lower_bound_0_1 $\epspop$]}. 
We first use the hypothesis \code{(g_nonneg env g)} as well as the nonnegativity of the rationals of the sequence $\lambda$ to deduce that \lstinline{eval_pol env (toPolC g $\lambda$ sos) $\geq$ 0}.
The nonnegativity of \lstinline{eval_pol env ($\epspop$ $\subcoq$ [lower_bound_0_1 $\epspop$])} comes from inequality~\eqref{eq:epspop}. 
\end{proof}
%\begin{lstlisting} \end{lstlisting}

We also define the type~\code{cert_pop_itv} of formal interval bounds certificates for polynomials expressions:
\begin{lstlisting}
Definition cert_pop_itv := (cert_pop * cert_pop).
\end{lstlisting}
Lemma~\code{correct_pop_itv} relates interval enclosure of polynomials with certificates of type~\code{cert_pop_itv}:
\if{\lstinline{$$}}\fi
\begin{lstlisting}
Lemma correct_pop_itv env g $\fpop$ ($i$ : itv) c: 
g_nonneg env g $\to$ checker_pop_itv g $\fpop$ $i$ c $\eqcoq$ true $\to$ $\evalexpr{\fpop}$ $\in$ $i$.
\end{lstlisting}

%\subsubsection{Formal Interval Bounds for Polynomials}
Here, \code{itv} refers to the type of intervals, encoded using two rational coefficients:
\begin{lstlisting}
Inductive itv : Type $\defcoq$ Itv : bigQ $\to$ bigQ $\to$ itv.
Definition itv01 $\defcoq$ Itv 0 1. (* the interval $[0, 1]$ *)
\end{lstlisting}
Moreover, we denote the lower (\resp{}upper) bound of an interval $i$ by $\underline{i}$ (\resp{}$\overline{i}$). 
We note \lstinline{$\evalexpr{p} \in i$} to state that \lstinline{$\underline{i}$ $\leq$ $\evalexpr{p}$ $\leq$ $\overline{i}$}.

\subsection{Experimental Results}
\label{subsect:popbenchs}
We first recall some Flyspeck related definitions~\cite{halesalgo}:
\small
\[
\begin{array}{rll}
\Delta \xb & := & x_1 x_4 ( x_2 -  x_1 +  x_3  - x_4 +  x_5 +  x_6)
+ x_2 x_5 (x_1  -  x_2 +  x_3 +  x_4  - x_5 +  x_6) \\
& & + x_3 x_6 (x_1 +  x_2  -  x_3 +  x_4 +  x_5  - x_6) 
 - x_3 ( x_2 x_4  + x_1 x_5)  -  x_6 (x_1 x_2  + x_4 x_5) \enspace, \\
\partial_4 \Delta \xb & := & x_1 ( x_2 -  x_1 +  x_3  - 2 x_4 +  x_5 +  x_6) 
 + x_2 x_5 + x_3 x_6 - x_3  x_2 - x_6 x_5  \enspace .
\end{array}
\]
\normalsize
We tested our formal verification procedure on the following polynomial problems, occurring as sub-problems of Flyspeck nonlinear inequalities:
\small
%see Appendix~\ref{ch:ineq} for the definition of  $\Delta \xb$):
%$\forall \xb \in [4, 6.3504]^3 \times [6.3504, 8] \times [4, 6.3504]^2, \partial_4 \Delta \xb \geq -41.$
%$\forall \xb \in [4, 6.3504]^3 \times [6.3504, 8] \times [4, 6.3504]^2, \Delta \xb \geq 0.$
\[
\begin{array}{rl}
\textit{POP1}: (4 \leq x_1, x_2, x_3, x_5, x_6 \leq 2.52^2 \wedge 2.52^2 \leq x_4 \leq 8) \implies &  \partial_4 \Delta \xb  \in  [-40.33, 40.33] \, . \\
\textit{POP2}: (4 \leq x_1, x_2, x_3, x_5, x_6 \leq 2.52^2 \wedge 2.52^2 \leq x_4 \leq 8) \implies  & 4 x_1 \Delta \xb  \in  [2047,  14262] \,.
\end{array}
\]
%\item {\em POP1}: \small $(4 \leq x_1, x_2, x_3, x_5, x_6 \leq 2.52^2 \wedge\ 2.52^2 \leq x_4 \leq 8) \implies \partial_4 \Delta \xb \geq -40.33 \enspace .$ \normalsize
%\item {\em POP2}: \small $(4 \leq x_1, x_2, x_3, x_5, x_6 \leq 2.52^2 \wedge\ 2.52^2 \leq x_4 \leq 8) \implies \Delta \xb \geq 2047 \enspace.$
\normalsize
A preliminary phase consists in scaling the POP to apply the correctness Lemma \code{correct_itv}.
Table~\ref{table:formalpop} shows some comparison results with the $\micromega$ tactic, available inside $\coq$. While performing the proof-search step, the tactic relies on the external SDP solver {\sc Csdp}. This solver is used to solve another SDP relaxation that is more general than $Q_k$ (Stengle Positivstellensatz~\cite{StenglePositiv}) to find witnesses of unfeasibility of a set of polynomial constraints (see ~\cite{Besson:2006:FRA:1789277.1789281} for more details).
%The tactic relies on the external SDP solver {\sc Csdp}.
To deal with the numerical errors of {\sc Csdp}, the proof-search ($\ocaml$ libraries) also includes a projection algorithm, which is is performed in such a way that $\epspop = 0$. Thus, the procedure returns a rational SOS certificate that matches exactly $\fpop - \tilde{\mu}_k$, so that the proof-checking consists only in verifying a polynomial equality.

Numerical experiments are performed using the $\coq$ proof scripts of our formalization (available in the $\nlcertify$~\footnote{\url{http://nl-certify.forge.ocamlcore.org/}} software package), on an Intel Core i5 CPU ($2.40\, $GHz). For the timings related to $\nlcertify$, the column ``$t_i$'' refers to the time spent to find the SOS certificates ``externally''  (while solving SDP relaxations and extracting SOS certificates in $\ocaml$) and the column ``$t_f$'' refers to the total verification time (while compiling $\coq$ proof scripts)\footnote{Note that obtaining $t_i$ while using $\micromega$ is possible in practice but would require to modify the $\ocaml$ libraries of the tactic.}.  Notice that for {\em POP2}, we consider the projection of $\Delta \xb$ with respect to the first $n$ coordinates on the box $K$ (fixing the other variables to $6.3504$). 

Table~\ref{table:formalpop} indicates that our tool outperforms the  $\micromega$ decision procedure, thanks to the sparse variant of relaxation $Q_k$ and a simpler projection method. The symbol ``--'' means that the inequality could not be checked by $\micromega$ within one hour of computation.  Problems occur while performing the proof-search step of $\micromega$, as either the projection algorithm fails or the computational cost of the SDP relaxation is too demanding.

%The main reason why \code{micromega} fails is th

\begin{table}[!ht]
\begin{center}
\caption{Comparing our formal POP checker with $\micromega$, }
\begin{tabular}{l|l|cc|c}
\hline
\multirow{2}{*}{Problem} & \multirow{2}{*}{$n$} & \multicolumn{2}{c|}{$\nlcertify$} & $\micromega$\\
& & $t_i$ & $t_f$ & $t_i + t_f$ \\
\hline
{\em POP1}           & 6 & $0.04 \, s$ & $0.16 \, s$ & $18 \, s$ \\
\hline
\multirow{2}{*}{{\em POP2}} & 2 & $0.06 \, s$ & $0.18 \, s$ & $0.72 \, s$ \\
                     & 3 & $0.14 \, s$ & $0.78 \, s$  & --  \\
                     & 6 & $4.8 \, s$ & $26.4 \, s$  & --  \\
\hline
\end{tabular}
\label{table:formalpop}
\end{center}
\end{table}

\if{
\begin{remark}%[Improving $\micromega$]
%The modularity of the $\micromega$ tactic allows several improvements, using our framework:
A part of our work could form the basis of an improved version of $\micromega$:
\begin{enumerate}
\item The modularity of the tactic  $\micromega$ allows to call the external library described in~\cite{Monniaux_Corbineau_ITP2011} to handle degenerate situations (when the SDP formulation of the POP is not strictly feasible). We interfaced our solver with this library and successfully solved small size instances of POP. The formal verification is much slower than the informal procedure (see the results presented in Section~\ref{sec:templates_bench}). As an example, for the {\em MC} problem,
it is $36$ times slower to generate exact SOS certificates and $13$
times slower to prove its correctness in $\coq$. This track was not further pursued, since computation were faster with the algorithms described in this section. However, we mention that this library could handle the formal verification of unconstrained POP, which is not possible with our current implementation.
%Note that the interface with $\coq$ still needs some streamlining.
\item The relaxation based on Stengle Positivstellensatz can be replaced by the sparse variant described in~\cite{Waki06sumsof}.
\end{enumerate}
\end{remark}
}\fi

\section{Formal Semialgebraic Optimization}
\label{sect:formalfsa}
We can now build on the work of the previous section in order to
extend the framework to obtain also formal bounds for semialgebraic optimization problems:
\begin{equation}
\label{eq:cons_fsa}
\fsa^*  :=  \inf_{\xb \in K} \fsa (\xb) \enspace,
\end{equation}
where $\fsa \in \setA$ and $K := \{\xb \in \R^n  :  g_1(\xb) \geq 0, \dots, g_m(\xb) \geq 0 \}$ is a basic semialgebraic set such that the constraints $(g_j)$ satisfy~\eqref{eq:norm_box}.
For the sake of clarity, we explicit only the subset of $\setA$ consisting of arbitrary composition of polynomials with 
 $\sqrt{\cdot}, +, -, \times, /$, whenever these operations are well-defined (neither division by zero nor square root of negative value occur). However, we can deal with the case $\fsa = \max (f_1, f_2)$ by using the identity: $2 \max (f_1, f_2) = f_1 + f_2 + \sqrt{(f_1 - f_2)^2}$. Similar identities exist to handle operations such as $\min (\cdot, \cdot), |\cdot|$. 
 
The function $\fsa$ has a basic semialgebraic lifting; this means that
one adds new ``lifting variables'' in order to get rid of the
non-polynomial functions in $\fsa$ thus reducing the problem to a POP (for more
details, see e.g. \cite{DBLP:journals/siamjo/LasserreP10}). More
precisely, 
we can add auxiliary variables $x_{n + 1},\dots,x_{n + p}$ (lifting variables),
and construct polynomials $ h_1, \dots , h_s \in \R[x_{1},\dots,x_{n + p}]$ defining the semialgebraic set: 
\[\Kpop := \{ (x_{1},\dots,x_{n + p}) \in \R^{n+p} : \xb \in K, 
 h_l(x_{1},\dots,x_{n + p}) \geq 0, \ l=1,\dots,s\} \: ,\]
 such that $\fpop^* := \inf \{ x_{n + p} : (x_{1},\dots,x_{n + p}) \in \Kpop\}$ is a lower bound of $\fsa^*$. 
 
Now, we explain how to implement this procedure in a formal setting.
\subsection{Data Structure for Semialgebraic Certificates}
\label{sect:coqfsa}
The inductive type \code{cert_sa} represents interval bounds certificates for semialgebraic optimization problems:
{  \small
\begin{lstlisting}
Inductive cert_sa : Type $\defcoq$
| Poly  : PExpr $\to$ itv $\to$ cert_pop_itv $\to$ cert_sa
| Fadd  : cert_sa $\to$ cert_sa  $\to$ itv $\to$ cert_pop_itv $\to$ cert_sa
| Fmul  : cert_sa $\to$ cert_sa $\to$ itv $\to$ cert_pop_itv $\to$ cert_sa
| Fdiv  : cert_sa $\to$ cert_sa  $\to$ itv $\to$ cert_pop_itv $\to$ cert_sa
| Fopp  : cert_sa $\to$ itv $\to$ cert_sa
| Fsqrt : cert_sa $\to$ itv $\to$ cert_sa.
\end{lstlisting}
}
Note that the constructor \code{Poly} takes a type \code{PExpr} object
as argument to represent the polynomial components of the function
$\fsa$. The other constructors correspond to the various ways of
building elements of $\setA$. Even though this certificate data-structure may look heavy-weighted, all constructors are required for verification purpose: for instance \code{Fadd} is mandatory to check the nonnegativity of a function such as $\sqrt{p} + q$, for polynomials $p$ and $q$.
%Given a semialgebraic function $\fsa$ and an interval enclosure $i$, a semialgebraic certificate encapsulates the interval $i$ and the expression of $\fsa$.
Each constructor takes a formal interval bound $i$ and a certificate
\code{c} of type \code{cert_pop_itv} (as defined in
section~\ref{sect:formalpop}) as arguments. In the sequel, we explain how to ensure that $i$ is a valid interval enclosure of $\fsa$ by checking the correctness of \code{c}.
%The function \lstinline{$\itvsos$ : cert_sa $\to$ itv} returns the interval used to build a semialgebraic certificate.
 
\if{
\begin{remark}
Here, we simplify the description of \code{cert_sa} for the sake of clarity. In our development, some constructors (\eg{}\code{Fadd}, \code{Fdiv}) take two intervals (denoted $i$ and $\isos$) as arguments. 
The interval $i$ refers to the result of an interval arithmetic operation, while $\isos$ is a formal interval bound of $f$, whose correctness is ensured with SOS certificates.
This second interval is tighter ($\isos \subseteq i$) but computing $i$ is mandatory for scaling purpose. 
\end{remark}
}\fi
%For illustration purpose, we will often refer to the following example:

\begin{example}[\thetheorem\ (from Lemma$_{9922699028}$ Flyspeck)]
\label{ex:fsa}
~\newline From the two multivariate polynomials $p(\xb) := \partial_4 \Delta \xb$ and $q(\xb) := 4 x_1 \Delta \xb$, we define the semialgebraic function $r(\xb) := p(\xb) / \sqrt{q(\xb)}$ over $K := [4, 2.52^2]^3 \times [2.52^2, 8] \times [4, 2.52^2]$. 
Using the procedure described in Section~\ref{sect:formalpop}, one obtains a formal interval $i_p$ (\resp{}$i_q$) enclosing the range of $p$ (\resp{}$q$), certified by an SOS certificate $c_p$ (\resp{}$c_q$). We also derive an interval enclosure $i_\surd := [m_7, M_7]$ for $\sqrt{q}$ and then build the following terms:
\begin{lstlisting}
Definition p := Poly $p$ $i_p$ $c_p$. Definition q := Poly $q$ $i_q$ $c_q$.
Definition sqrtq := Fsqrt q $i_{\surd}$.
Definition r := Fdiv p sqrtq $i_r$ $c_r$. 
\end{lstlisting}
The SOS certificate $c_r$ allows to prove that $i_r$ is a correct interval enclosure of $r$. 
%The interval $i_r := [m_8, M_8]$ is obtained using formal interval arithmetic division $i_p \: \hat{/} \: i_\surd$.
\end{example}

The interpretation of~\code{cert_sa} objects is straightforward, using the evaluation procedure for polynomial expressions. Thus, we also note \lstinline{$\evalexpr{f}$} the interpretation of the semialgebraic certificate $f$, which returns the expression of the semialgebraic function $\fsa$.
A procedure \lstinline{$\liftcoq$} returns the index \code{v} of the lifting variable which represents $\fsa$. 
%and takes a nonnegative counter \code{v} as argument. 
When $\fsa$ is a polynomial, \lstinline{$\liftcoq$} returns the number of variables involved in $\fsa$. The value is incremented when $\fsa$ is either a division or a square root.
\begin{lstlisting}
Fixpoint $\liftcoq$ $f$ v $\defcoq$
  match $f$ with
  | Poly $p$ _ _ $\implies$ v   | Fopp $a$ _ $\implies$ $\liftcoq$ $q$ v
  | Fdiv $f_1$ $f_2$ _ _ $\implies$ $\liftcoq$ $f_2$ ($\liftcoq$ $f_1$ v) + 1
  | Fsqrt $q$ _ $\implies$  $\liftcoq$ $q$ v + 1
  | Fadd $f_1$ $f_2$ _ _ | Fsub $f_1$ $f_2$ _ _ | Fmul $f_1$ $f_2$ _ _ $\implies$ $\liftcoq$ $f_2$ ($\liftcoq$ $f_1$ v)
  end.
\end{lstlisting}
Then, two procedures are mandatory to reduce Problem~\eqref{eq:cons_fsa} into a polynomial optimization problem $\min_{\xb \in \Kpop} \fpop$. The function \code{obj} derives the objective polynomial $\fpop$, while \code{cstr} returns a list of polynomials $(h_l)$ defining $\Kpop$. 

\begin{lstlisting}
Fixpoint obj $f$ v : PExpr := 
  match $f$ with
   | Poly $p$ _ _ $\implies$ $p$
   | Fopp $q$ _ $\implies$ - (obj $q$ v)
   | Fadd $f_1$ $f_2$ _ _ $\implies$ (obj $f_1$ v) + (obj $f_2$ ($\liftcoq$ $f_1$ v))
   | Fmul $f_1$ $f_2$ _ _ $\implies$ (obj $f_1$ v) * (obj $f_2$ ($\liftcoq$ $f_1$ v))
   | Fsub $f_1$ $f_2$ _ _ $\implies$ (obj $f_1$ v) - (obj $f_2$ ($\liftcoq$ $f_1$ v))
   | Fdiv $f_1$ $f_2$ $i$ _ $\implies$ 
            scale_obj (PEX ($\liftcoq$ $f_2$ ($\liftcoq$ $f_1$ v) + 1)) $i$
   | Fsqrt $q$ $i$ $\implies$ scale_obj (PEX ($\liftcoq$ $q$ v + 1)) $i$
  end.
\end{lstlisting}
Given a polynomial $p$ and an interval $i := [m, M]$, the result of (\code{scale_obj} $p$ $i$) is $(M - m) p + m$.

\begin{lstlisting}
Fixpoint cstr $f$ v : seq PExpr := 
  match $f$ with 
   | Poly $p$ _ _ $\implies$ [::]
   | Fopp $q$ _ $\implies$ cstr $q$ v
   | Fadd $f_1$ $f_2$ _ _ | Fsub $f_1$ $f_2$ _ _ | Fmul $f_1$ $f_2$ _ _ => cstr $f_1$ v ++ cstr $f_2$ (var $f_1$ v)
   | Fdiv $f_1$ $f_2$ _ _ $\implies$ cstr $f_1$ v ++ cstr $f_2$ (var $f_1$ v) ++ [::PEsub (PEmul (obj $f_2$ (var $f_1$ v)) (obj $f$ v)) (obj $f_1$ v); PEsub (obj $f_1$ v) (PEmul (obj $f_2$ (var $f_1$ v)) (obj $f$ v))]
   | Fsqrt $q$ _ $\implies$ cstr $q$ v ++ [::PEsub (PEmul (obj $f$ v) (obj $f$ v)) (obj $q$ v) ; PEsub (obj $q$ v) (PEmul (obj $f$ v) (obj $f$ v))]
  end.
\end{lstlisting}

\begin{example}
\label{ex:fsanext}
Applying the function \lstinline{$\liftcoq$} to the six dimensional functions $q$ and $r$ of Example~\ref{ex:fsa} yields \lstinline{$\liftcoq$ sqrtq $\eqcoq$ $7$} and \lstinline{$ \liftcoq$ r $\eqcoq$ $8$}. 
Then, \lstinline{obj sqrtq} returns the polynomial $(M_7 - m_7) \ x_7 + m_7$ and \lstinline{obj r} returns $(M_8 - m_8) \ x_8 + m_8$. 
The interval $[m_8, M_8]$ is obtained using formal interval arithmetic division $i_p \: \hat{/} \: i_\surd$.
This scaling procedure allows to use the function \code{lower_bound_0_1} since one can prove that $x_7, x_8 \in [0, 1]$. Here \lstinline{cstr sqrtq} returns the finite sequence of polynomials $l_\surd := [h_1; h_2; h_3; h_4]$, with $h_1 := [(M_7 - m_7) \ x_7 + m_7]^2 - q$, $h_2 := - h_1$, $h_3 := x_7$ and $h_4 := 1 - x_7$. Next, \lstinline{cstr r} returns the concatenation of $l_\surd$ with $l_r := [h_5; h_6; h_7; h_8]$, where $h_5 := [(M_8 - m_8) \ x_8 + m_8] [(M_7 - m_7) \ x_7 + m_7] - p$, $h_6 := - h_5$, $h_7 := x_8$ and $h_8 := 1 - x_8$.
\end{example}

\if{
We define the inclusion relation \lstinline{(Itv a b)  $\subseteq$ (Itv a' b')} whenever $a' \leq a$ and $b \leq b'$.
In the sequel, we use the basic operations of interval arithmetics.

\begin{itemize}
\item Addition: \lstinline{Itv $a_1$ $b_1$ $\hat{+}$ Itv $a_2$ $b_2$ $\eqcoq$ Itv ($a_1 + a_2$) ($b_1 + b_2$)}
\item Opposite: \lstinline{$\hat{-}$ Itv $a$ $b$ $\eqcoq$ Itv $\subcoq a$ $\subcoq b$} and subtraction: $i_1 \hat{-} i_2 = i_1 \hat{+} (\hat{-} i_2)$
\item Multiplication: \lstinline{Itv $a_1$ $b_1$ $\hat{\times}$ Itv $a_2$ $b_2$ $\eqcoq$ Itv ($\min C$) ($\max C$)}, where $C := \{ a_1 b_1, a_2 b_1, a_1 b_2, a_2 b_2 \}$
\item Inverse: \lstinline{$\hat{/}$ Itv $a$ $b$ $\eqcoq$ Itv $1 / a$ $1 / b$} and division: $i_1 \hat{/} i_2 = i_1 \hat{\times} (\hat{/} i_2)$
\item Square: \lstinline{(Itv $a$ $b$)$^{\hat{2}}$ $\eqcoq$ Itv $(\max \{0, a, -b \})^2$ $(\max \{- a, b \})^2$}
\end{itemize}
}\fi

\subsection{Formal Interval Bounds for Semialgebraic Functions}
\label{subsect:formalfsa}

Now, we introduce the function \code{checker_sa} built on top of \code{checker_pop_itv}, which checks recursively the correctness of certificates for semialgebraic functions.  
For the sake of simplicity and to stay consistent with Example~\ref{ex:fsanext}, we only present the result of the procedure for the constructors \code{Poly}, \code{Fdiv} and \code{Fsqrt}. We use the inclusion relation: $[a, b]\subseteq [a', b']$ whenever $a' \leq a$ and $b \leq b'$, the formal interval arithmetic square: $\sq \: [a, b] := [(\max \{0, a, -b \})^2, (\max \{- a, b \})^2]$, as well as the interval positivity: $[a, b] \: \hat{>} \: 0$ whenever $a > 0$.
\begin{lstlisting}
Fixpoint checker_sa g $f$ : bool $\defcoq$ 
  match $f$ with 
  | Poly $p$ $i$ c $\implies$ checker_pop_itv g $p$ $i$ c
  | Fdiv $f_1$ $f_2$ $i$ c $\implies$ 
      checker_sa g $f_1$ && checker_sa g $f_2$ && $0$ $\notin$ ($\itvsos$ $f_2$) 
   && checker_pop_itv (g ++ cstr $f$) (obj $f$) $i$ c
  | Fsqrt $q$ $i$ $\implies$ checker_sa g $q$ &&  
    $\itvsos$ $q$ $\hat{>}$ 0 && $\underline{i} < \overline{i}$ &&  ($\itvsos$ $q$) $\subseteq$  $\sq$ $i$
  ...
  end. 
\end{lstlisting}
%When $f$ is a polynomial $p$, the checker verifies that 
Recall that our goal is to prove the correctness of a lower bound $\mu^-$  of $\min_{\xb \in K} \fsa (\xb)$ (\resp{}an upper bound $\mu^+$ of $\max_{\xb \in K} \fsa (\xb)$). Then, one can apply the following correctness lemma to state that the interval $i := [\mu^-, \mu^+]$ (returned by \lstinline{$\itvsos$ $f$}) is a valid enclosure of $\fsa$ over $K$ whenever one succeeds to check the certificate $f$.
\begin{lstlisting}
Lemma correct_fsa env g $f$ v : g_nonneg env g $\to$
checker_sa g $f$ $\eqcoq$ true $\to$ $\evalexpr{f}$ $\in$ ($\itvsos$ $f$).
\end{lstlisting}
\begin{proof}
By induction over the structure of semialgebraic expressions.
\end{proof}

\begin{example}[\thetheorem\ (Formal bounds for the function of Example~\ref{ex:fsa})]
\label{ex:fsaend}
~\newline
Continuing Example~\ref{ex:fsanext}, one considers the POP:
\[\min_{\xb, x_7, x_8} \{(M_8 - m_8) \ x_8 + m_8) : \xb \in K, h_1(\xb, x_7, x_8) \geq 0, \dots, h_8 (\xb, x_7, x_8) \geq 0 \} \: ,\] 
to bound from below the function $r(\xb) := p(\xb) / \sqrt{q(\xb)}$. Solving this POP using the second order SDP relaxation $Q_2$ yields the lower bound $\mu_2^- = -0.618$. Similarly, one obtains the upper bound $\mu_2^+ = 0.892$. The procedure \code{checker_sa} calls the function \code{checker_pop_itv} to prove the correctness of the interval bounds $i_p$, $i_q$ (as detailed in Section~\ref{subsect:popbenchs}) and $i_r := [\mu_2^-, \mu_2^+]$. The total running time of this formal verification in $\coq$ is about $200 s$.  Adding the bit-size of all rational coefficients involved in this certificate yields a total of about 667 kbit. About $90 \%$ of the CPU time is spent verifying the correctness  of SOS certificates, that is checking polynomial equalities with the \code{ring} tactic. 

The corresponding proof script is available in the $\nlcertify$ package~\footnote{the file \code{coq/fsacertif.v} in the archive at~\url{https://forge.ocamlcore.org/frs/?group_id=351}}.
\end{example}

\section{Certified Bounds for Multivariate Transcendental Functions}
\label{sect:transc}

We now consider an instance of Problem~\eqref{eq:f}. We identify the objective function $f$ with its
abstract syntax tree $t$, whose leaves are
semialgebraic functions (see Section~\ref{sect:formalfsa}) and other nodes are either basic binary
operations ($+$, $\times$, $-$, $/$) or belong to the set $\setD$ of unary transcendental functions ($\sin$, \etc{}).
We first recall how to handle these unary functions using maxplus approximations.

\subsection{Maxplus Approximations for Univariate Semiconvex Transcendental Functions}
\label{subsect:maxplus}
We consider transcendental functions which are twice differentiable. Thus, the restriction of $r \in \setD$ to any closed interval $I$ is $\gamma$-semiconvex for a sufficiently large $\gamma$, \ie{}the univariate function $g := r + \frac{\gamma}{2} |\cdot|^2$ is convex on $I$ (for more details on maxplus approximation, we refer the interested reader to~\cite{agk04, mceneaney-livre}). Using the convexity of $g$, one can always find a constant $\gamma \leq \sup_{b\in I} -r''(b)$ such that for all $b_i \in I$:
\begin{align}
\label{eq:maxplus}
\forall b \in I, \quad r (b)  \geq \parab_{b_i}^-(b):=  -\frac{\gamma}{2} (b-b_i)^2 +r'(b_i) (b - b_i) + r (b_i) \enspace .
\end{align}
Note that the choice $\gamma = \sup_{b\in I} -r''(b)$ is always valid.
By selecting a finite subset of control points $B \subset I$, one can bound $r$ from below using a maxplus under-approximation:
\begin{align}
\label{eq:maxplusmax}
\forall b \in I, \quad r (b)  \geq \max_{b_i \in B} \parab_{b_i}^-(b) \enspace .
\end{align}

\begin{example}%[Modified Schwefel Problem]
\label{ex:swf}
Consider the function $f := \sum_{i=1}^{n } x_i \sin (\sqrt{x_i})$ defined over $[1, 500]^n$ (Modified Schwefel Problem~\cite{Ali:2005:NES:1071327.1071336}) and the finite set $\{b_1, b_2, b_3\}$ of control points, with $b_1 := 135$, $b_2 := 251$, $b_3 := 500$.
For each $i=1,\dots,n$, consider the sub-tree
$\sin(\sqrt{x_i})$. 
First, we get the equations of
$\parab_{b_1}^-$, $\parab_{b_2}^-$ and $\parab_{b_3}^-$, which are three under-approximations of the 
function $b \mapsto \sin(\sqrt{b})$ on the real interval $I := [1, \sqrt{500}]$. Similarly we
obtain three over-approximations $\parab_{b_1}^+$, $\parab_{b_2}^+$ and
$\parab_{b_3}^+$ (see Figure~\ref{fig:sin_sqrt3}). Then, we obtain the under-approximation $t_i^- :=
\max_{j \in \{1, 2, 3\}} \{ \parab_{b_j}^- (x_i)\}$ and the
over-approximation $t_i^+ := \min_{j \in \{1, 2, 3\}}
\{ \parab_{b_j}^+ (x_i)\}$. 
\begin{figure}[!t]
\begin{center}
\includegraphics[scale=1]{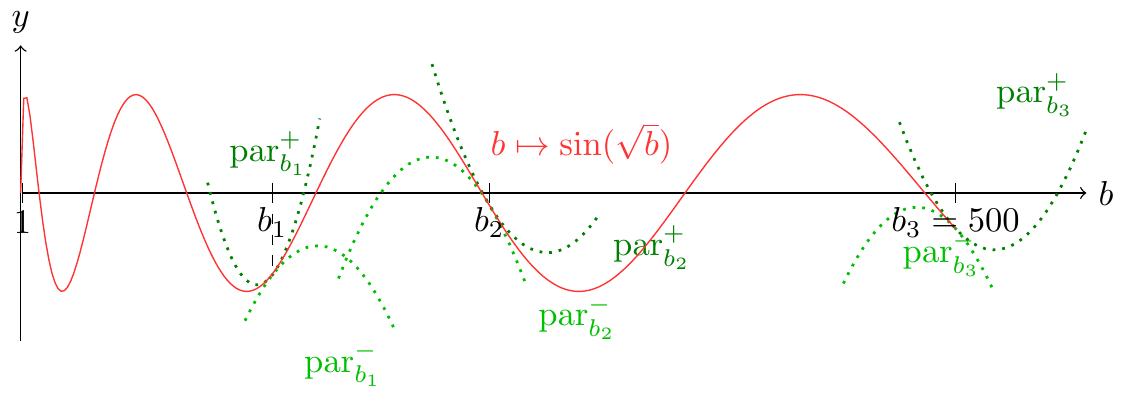}
\caption{Template Maxplus Semialgebraic Approximations for $b \mapsto \sin(\sqrt{b})$: $\max_{j \in \{1, 2, 3\}} \{ \parab_{b_j}^- (x_i)\} \leq  \sin{\sqrt{x_i}} \leq \min_{j \in \{1, 2, 3\}} \{ \parab_{b_j}^+ (x_i)\} $}\label{fig:sin_sqrt3}
%\caption{A hierarchy of Semialgebraic Approximations for $b \mapsto \sin(\sqrt{b})$}	\label{fig:sin_sqrt3}
\end{center}
\end{figure}
\end{example}

\subsection{The Nonlinear Maxplus Template Method}
\label{subsect:template}

Our main algorithm $\templateapprox$ (Figure~\ref{alg:templateapprox}) is based on a previous method of the authors~\cite{victormpb}, in which the objective function is bounded by means of semialgebraic functions. For the sake of completeness, we first recall the basic principles of this method.

Given a function represented by an abstract tree $t$, semialgebraic lower and upper approximations $t^-$ and $t^+$ are computed by induction.
If the tree is reduced to a leaf, \ie{}$t \in \setA$, we set $t^- = t^+ := t$. If the root of the tree corresponds to a binary operation $\bop$ with children $c_1$ and $c_2$, then the semialgebraic approximations $c_1^-$, $c_1^+$ and $c_2^-$, $c_2^+$ are  composed using a function $\composebop$ to provide bounding approximations of $t$. Finally, if $t$ corresponds to the composition of a transcendental (unary) function $r$ with a child $c$, we first bound $c$ with semialgebraic functions $c^+$ and $c^-$. We compute a lower bound $c_m$ of $c^-$ as well as an upper bound  $c_M$ of $c^+$ to obtain an interval $I := [c_m, c_M]$ enclosing $c$. Then, we bound $r$ from below and above by computing parabola at given control points with a function called $\buildpar$, thanks to the semiconvexity properties of $r$ on the interval $I$ (e.g. the functions $r^- := \max_{j \in \{1, 2, 3\}} \{ \parab_{b_j}^- \}$ and $r^+ := \min_{j \in \{1, 2, 3\}} \{ \parab_{b_j}^+\}$ from Example~\ref{ex:swf}).
These parabola are composed with $c^+$ and $c^-$, thanks to a function denoted by $\composeapprox$ (Line~\lineref{line:compose_template}). 

At the end (Line~\lineref{line:minsa}), we call the function $\minsa$ (\resp{}$\maxsa$) which determines lower (\resp{}upper) bounds of the approximation $t^-$ (\resp{}$t^+$) using techniques presented in Section~\ref{sect:formalfsa}.
\begin{figure}[!ht]
\begin{algorithmic}[1]                    
\Require tree $t$, semialgebraic set $K$, finite sequence of control points $s$
\Ensure lower bound $m$,  upper bound $M$, lower semialgebraic approximation $t^-$, upper semialgebraic approximation $t^+$  
	\If {$t \in \setA$}\label{line:begin_template} $t^- := t$, $t^+ := t$
	\ElsIf {$\bop := \astroot (t)$ is a binary operation with children $c_1$ and $c_2$}
		\State $m_i, M_i, c_{i}^-, c_{i}^+ := \templateapproxjfr{c_i}{K}{s}$ for $i \in \{1,2\}$
		\State $t^-, t^+ := \composebop (c_{1}^-, c_{1}^+, c_{2}^-, c_{2}^+, \bop)$ \label{line:samp_template1}
	\ElsIf {$r := \astroot (t)$ is a univariate transcendental function with a child $c$}	
		\State $m_c$, $M_c$, $c^-$, $c^+ := \templateapproxjfr{c}{K}{s}$ \label{line:samp_template}
		\State $I := [m_{c}, M_{c}]$
		\State $r^-, r^+ := \buildpar (r, I, c, s)$\label{line:buildpar}
		\State $t^-, t^+ :=  \composeapprox (r, r^-, r^+, I, c^-, c^+)$\label{line:compose_template}
	\EndIf\label{line:end_template}
	\State \Return $\minsa(t^-, K)$, $\maxsa(t^+, K)$, $t^-$, $t^+$ \label{line:minsa}
\end{algorithmic}
\caption{$\templateapprox$}
\label{alg:templateapprox}
\end{figure}

\begin{example}
\label{ex:swfnext}
We illustrate the nonlinear maxplus template method with the function $f$ of~Example~\ref{ex:swf}. 
We approximate $f$ with maxplus approximations built with 3 control points (Figure~\ref{fig:sin_sqrt3}), which allows to reduce the modified Schwefel problem to the following POP:
\[
\left\{			
\begin{array}{rl}
\min\limits_{\xb, \zb} 
& - \sum_{i=1}^{n} x_i z_i \\			 
\text{s.t.} & z_i \leq \parab_{b_j}^+(x_i) \:, \quad j \in \{1, 2, 3\}\:,\quad  i = 1,\dots,n \enspace, \\
& \xb \in [1, 500]^n\:, \quad \zb \in [-1, 1]^n \enspace.
\end{array} \right.
\]
\end{example}
%TODO which are computed by \minsa and \maxsa
\subsection{Formal Verification of Semialgebraic Relaxations}
\label{subsect:interface}
The correctness of the semialgebraic maxplus approximations for univariate functions (computed by the $\buildpar$ procedure, see Figure~\ref{alg:templateapprox} at Line~\lineref{line:buildpar}) is ensured with the $\coqinterval$~\footnote{\url{https://www.lri.fr/~melquion/soft/coq-interval/}} tactic~\cite{Melquiond201214}, available inside $\coq$. As detailed in Section~\ref{subsect:formalfsa}, the procedure \code{checker_sa} validates the interval bounds for semialgebraic problems obtained with the functions $\minsa$ and $\maxsa$ (see Figure~\ref{alg:templateapprox} at Line~\lineref{line:minsa}).

Table~\ref{table:formalfly} presents the results obtained when proving the correctness of lower bounds for semialgebraic relaxations of two 6-variables Flyspeck inequalities. When the bounds obtained with the algorithm $\templateapprox$ are too coarse to certify a given inequality, we perform a branch and bound procedure over the domain $K$. We refer to $\nbb$ as the total number of domain cuts that are mandatory to prove the inequality.  As for Table~\ref{table:formalpop}, the time $t_i$ refers to the informal verification time, required to construct the certificates for semialgebraic functions, while using the optimization algorithm without any call to the $\coq$ libraries. The total verification time $t_f$ is then compared with $t_i$.
% The execution time is compared with the time which is required to construct the SOS certificates (informal verification time). 

\begin{table}[!ht]
\begin{center}
\caption{Formal Bounds Computation Results for Semialgebraic Relaxations of Flyspeck Inequalities}
\begin{tabular}{lcccc}
\hline
Inequality  & $\nbb$  & $t_i$ & $t_f$ & $\frac{t_i + t_f}{t_i}$ \\
%& &  Optimization Time &  Optimization Time\\
\hline
$9922699028$% & 14 & $244 \, s$ & $972 \, s$ \\
                              & 39 & $295 \, s$ & $2218 \, s$ & 8.5\\
\hline
$3318775219$ %& $\minimaxoptim$ & 266 & $4423 \, s$ & $18255 \, s$ \\
                              & 338 & $2285 \, s$ & $19136 \, s$ & 9.4\\
\hline
\end{tabular}
\label{table:formalfly}
\end{center}
\end{table}
Here, the formal verification of SOS certificates is the bottleneck of the computational certification task. Indeed, it is 8.5 (resp.~9.4) times  slower to prove the correctness of semialgebraic lower bounds for the first (resp.~second) inequality.  For both inequalities, it takes about 7\% of the total time to compute bounds with SDP. Note that half this time is spent to compute negative bounds which are not formally verified afterwards. Such non trivial inequalities are also used as test cases for the formal techniques employed by the Flyspeck project (see the row corresponding to the inequality 3318775219 in Table 2 of~\cite{DBLP:journals/corr/abs-1301-1702}) and it takes about the same amount of CPU time to verify them with both methods. For comparison purpose, notice that this ratio between formal and informal verification does not exceed $10$ in our case, while it is about 2000$\sim$4000 in~\cite{DBLP:journals/corr/abs-1301-1702}. Also, as mentioned in~\cite{victormpb}, the number of subdivisions is much smaller than for methods using interval Taylor approximation (9370 for the first inequality and 25994 for the second one\footnote{The data come from the benchmarks file available at~\url{http://code.google.com/p/flyspeck/source/browse/trunk/informal_code/interval_code/qed_log_2012.txt}. The file indicates the number of subdivisions (denoted by ``cells'') for each inequality while running informal verification in C++ with interval arithmetic and directed rounding.}), due to the precision of SOS-based methods.

Table~\ref{table:formalgo} presents the results obtained for examples issued from the global optimization literature (see Appendix B in~\cite{Ali:2005:NES:1071327.1071336} for more details). For each problem, we indicate the number of subdivisions $\nbb$ that are performed to obtain the lower bound $m$ with our method.

\begin{example}
We recall the definition of Problem{(MC)}:\newline
$\min_{\xb \in [-1.5, 4] \times [-3, 3]} \sin (x_1 + x_2) + (x_1 - x_2)^2 - \frac{3}{2} x_1 + \frac{5}{2} x_2 + 1$.

The package $\nlcertify$ contains an example of proof obligations for Problem~{\it MC}~\footnote{the file \code{coq/mccertif.v} in the archive at~\url{https://forge.ocamlcore.org/frs/?group_id=351}} on the box $[- \frac{3}{2}, - \frac{1}{8}] \times [-3, -\frac{9}{4}]$, allowing to assert the following:
\begin{lemma}
\label{th:mc}
$\forall x_1, x_2 \in [0, 1], -1.92 \leq \sin (\frac{11}{8} x_1 - \frac{3}{2} + \frac{3}{4} x_2 - 3) +
   (\frac{11}{8} x_1 - \frac{3}{2} + \frac{3}{4} x_2 - 3)^2  -
    \frac{3}{2} (\frac{11}{8} x_1 - \frac{3}{2}) + \frac{5}{2} (\frac{3}{4} x_2 - 3) +1.$
\end{lemma}
\begin{proof}
Using the interval tactic, one proves that 

\begin{enumerate}[noitemsep,topsep=2pt,label={(\roman*)}]
\item $\forall z \in [- \frac{9}{2}, - \frac{19}{8}], r^-(z)  \leq \sin z$, where the parabola $r^-$ is defined as follows:
\newline
{\footnotesize $$  r^-(z) := - \frac{68787566775937}{140737488355328} z^2 - \frac{104740403727667521893}{23346660468288651264} z - \frac{145294742556168586619925337}{15491723247053871123529728} \:.$$}
\item By substitution, it follows that $\forall x_1, x_2 \in [0, 1], r^-(\frac{11}{8} x_1 - \frac{3}{2} + \frac{3}{4} x_2 - 3) \leq \sin (\frac{11}{8} x_1 - \frac{3}{2} + \frac{3}{4} x_2 - 3)$.
\item Using a Putinar-type certificate, one checks the nonnegativity of the left-hand side polynomial with the procedure \code{correct_pop}, which yields the desired result. % and by transitivity, one gets the desired result.
\end{enumerate}
\end{proof}
\end{example}
\begin{table}[!t]
\begin{center}
\caption{Formal Bounds Computation Results for Semialgebraic Relaxations of Global Optimization Problems}
\begin{tabular}{lcccccc}
\hline
Problem  & $n$ & $m$ & $\nbb$ & $t_i$ & $t_f$ & $\frac{t_i + t_f}{t_i}$\\
%& & & &  Optimization Time &  Optimization Time\\
\hline
{\em MC} & $2$ &  $-1.92$&  $17$  & $1.8 \, s$ & $1.9 \, s$ & 2.1 \\
\hline
{\em SWF} & $5$ &  $-2150$&  $78$  & $270 \, s$ & $477 \, s$ & 2.8\\
\hline
\end{tabular}
\label{table:formalgo}
\end{center}
\end{table}

\section{Conclusion}
This framework allows to prove formal bounds for nonlinear optimization problems.  The SOS certificates checker benefits from a careful implementation of informal and formal libraries.
The informal certification tool exploits the system properties of the problems to derive semialgebraic relaxations involving less SOS variables, thus more concise certificates. Our simple projection procedure yields  SOS polynomials with arbitrary-size rational coefficients, that are efficiently checked on the $\coq$ side, thanks to the machine modular arithmetic.
%by defining a customized ring of polynomials with arbitrary-size rational coefficients.
The formal libraries can currently verify medium size semialgebraic certificates  for global optimization problems and inequalities arising in the proof of Kepler's Conjecture. The implementation of polynomial arithmetic still needs some streamlining, as checking ring equalities in $\coq$ remains the bottleneck of our verification procedure. A possible workaround to handle larger size problems is to use polynomials with interval coefficients as in~\cite{brisebarre:ensl-00653460}, so that one could obtain formal bounds without computing the exact polynomial remainder $\epspop$. 
We plan to complete the formal verification procedure
by additionally  automatizing in $\coq$ the proof of correctness of the maxplus semialgebraic
approximations.
A topic of further investigation is to evaluate the resulting improved methodology on all Flyspeck inequalities as well as on the sample of global optimization problems informally solved in~\cite{victormpb}.

%The implementation of polynomial arithmetic still needs some streamlining, as checking ring equalities in $\coq$ remains the bottleneck of our verification procedure. 

\newcommand{\etalchar}[1]{$^{#1}$}

%\bibliographystyle{alpha}
%\bibliography{formal_templates}

\end{document}